\documentclass[draftcls,onecolumn,12pt]{IEEEtran}
\usepackage{color}
\usepackage{verbatim}
\usepackage{amsfonts}
\usepackage{amssymb}
\usepackage{stfloats}
\usepackage{cite}
\usepackage{graphicx}
\usepackage{psfrag}
\usepackage{subfigure}
\usepackage{amsmath}
\usepackage{array}
\usepackage{epstopdf}
\usepackage{authblk}
\usepackage{graphicx} 
\usepackage{amsthm} 
\usepackage{lipsum}
\usepackage{verbatim} 
\usepackage{authblk}
\usepackage{mathtools}
\usepackage{cuted}
\usepackage[ruled,linesnumbered]{algorithm2e}
\usepackage{algpseudocode}
\usepackage{framed} 
\usepackage{subfigure}
\usepackage{soul}
\usepackage{bm}
\usepackage{setspace}
\usepackage{dsfont}
\usepackage[breaklinks,colorlinks,linkcolor=black,citecolor=black,urlcolor=black]{hyperref}
\newtheorem{theorem}{Theorem}

\newtheorem{proposition}{Proposition}

\newtheorem{remark}{\bf Remark}

\def\qed{$\Box$}

\def\phi{\varphi}

\def\({\left(}
\def\){\right)}

\def\b0{{\mathbf{0}}}

\def\cF{\mathcal{F}}
\def\cG{\mathcal{G}}

\def\cK{\mathcal{K}}

\def\cS{\mathcal{S}}

\definecolor{LatestRevision}{rgb}{0.32, 0.18, 0.5}
\allowdisplaybreaks[4]

\title{Joint Batching and Scheduling for High-Throughput Multiuser Edge AI with Asynchronous Task Arrivals
}

\author{Yihan~Cang, \emph{Graduate Student Member, IEEE}, Ming Chen, \emph{Member, IEEE},\\ and Kaibin~Huang, \emph{Fellow, IEEE} 
	
\thanks{
		
Y. Cang is with Department of Electrical and Electronic Engineering at The University of Hong Kong, Hong Kong, and also with National Mobile Communications Research Laboratory, Southeast University, Nanjing 211111, China (email: yhcang@hku.hk). 

M. Chen is with National Mobile Communications Research Laboratory, Southeast University, Nanjing 211111, China, and also with Purple Mountain Laboratories, Nanjing 211100, China (email: chenming@seu.edu.cn).

K. Huang is with Department of Electrical and Electronic Engineering at The University of Hong Kong, Hong Kong (email: huangkb@eee.hku.hk). Corresponding author: K. Huang.}}

\makeatletter
\newcommand{\removelatexerror}{\let\@latex@error\@gobble}
\makeatother

\IEEEoverridecommandlockouts

\begin{document}

\maketitle
\vspace{-10mm}
\begin{abstract}
    Edge \emph{artificial intelligence} (AI) in the sixth-generation networks will provide inference services at the network edge to enrich the capabilities of mobile devices and lengthen their battery lives. As a well-known technique in computing, batching can boost the computation throughput at an edge server by assembling multiple tasks into a batch that is fed into a pre-trained prediction model. This reduces the memory-access frequency and hence accelerates the execution of each task. In a multiuser edge-AI system, the end-to-end latency depends not only on computation but also on communication, i.e., multiuser task uploading over a multi-access channel. In this paper, we study joint batching and (task) scheduling to maximise the throughput (i.e., the number of completed tasks) under the practical assumptions of heterogeneous task arrivals and deadlines. The design aims to optimise the number of batches, their starting time instants, and the task-batch association that determines batch sizes. The joint optimisation problem is complex due to multiple coupled variables as mentioned and numerous constraints including heterogeneous tasks arrivals and deadlines, the causality requirements on multi-task execution, and limited radio resources. Underpinning the problem is a basic tradeoff between the size of batch and waiting time for tasks in the batch to be uploaded and executed. Our approach of solving the formulated  mixed-integer problem is to transform it into a convex problem via integer relaxation method and $\ell_0$-norm approximation. This results in an efficient alternating optimization algorithm for finding a close-to-optimal solution. Specifically, it iterates between solving two sub-problems, optimal task-batch association and optimal batch starting time. The former is a linear program whose solution can be found using a derived scheme of greedy task selection while that of the latter is derived in closed form. In addition, we also design the optimal algorithm from leveraging \emph{spectrum holes}, which are caused by fixed bandwidth allocation to devices and their asynchronized multi-batch task execution, to admit unscheduled tasks so as to further enhance throughput. Simulation results demonstrate that the proposed framework of  joint batching and resource allocation can substantially enhance the throughput of multiuser edge-AI as opposed to a number of simpler benchmarking schemes, e.g., equal-bandwidth allocation, greedy batching and single-batch execution.
\end{abstract}


\section{Introduction}
\label{sec: introduction}

Edge \emph{Artificial Intelligence} (AI), a key feature of the sixth-generation (6G) mobile networks, will feature ubiquitous deployment of AI algorithms at the network edge to provide inference services to users \cite{9606720,8970161}. Then Internet-of-Things (IoT) devices can rely on the services to acquire intelligent capabilities ranging from visual perception to natural language processing. Realizing efficient edge AI in practice has to overcome both the communication and computing bottlenecks. The former results from many devices uploading high-dimensional data features to an edge server over a resource constrained multi-access channel. The second refers to the well known \emph{von Neumann bottleneck} where frequent data shuttling  between memory and processors (e.g., loading of AI model parameters) can incur as much as $90\%$ of total computation latency and energy \cite{backus1978can, zou2021breaking}. The consideration of end-to-end system performance makes it important to simultaneously overcome the two bottlenecks, which motivates this work. To this end, we design a framework of integrating batching (i.e., task execution in batches to alleviate the von Neumann bottleneck) and device scheduling to enhance the throughput of an edge AI system under the practical assumptions of heterogeneous task arrivals and deadlines. 

The area of edge AI, also called edge inference, involves cross-disciplinary research integrating wireless communication and AI to improve the end-to-end system performance \cite{9606720,9311935}. Many relevant algorithms are designed based on a popular architecture called \emph{split inference} that partitions a global deep neural network into an on-device and a server sub-models, which are connected by a wireless channel \cite{9311935, 9144210}. Given the architecture, a rich set of techniques have been designed to improve the communication efficiency including pruning the features extracted using the on-device sub-model \cite{Niu2019Infocom,Deniz2020SPAWC}, jointly training the sub-model and channel encoder \cite{9311935},  progressive transmission \cite{9955582}, and distributed data compression using the information-bottleneck approach \cite{9837474,9606667}. Controlling the model splitting point for split inference introduces another dimension for improving the communication efficiency. In \cite{9144210}, the point is jointly optimized with computation-resource allocation for a multi-core CPU to minimize the end-to-end latency of multiuser  tasks. From the perspective of implementation, edge AI algorithms can be deployed  on the \emph{mobile edge computing} (MEC) platform, a focus of 5G development,  to exploit its strengths in enabling latency-critical applications such as virtual reality  (see, e.g., \cite{8876870}). Furthermore, various practical issues for edge AI deployment have been addressed by researchers such as  joint management of communication and computation resources (see, e.g., \cite{9843917}), heterogeneous devices \cite{9296560}, and random task arrivals (see, e.g., \cite{9795664}).

In the context of multiuser edge AI, batching is mentioned earlier to be an effective technique for breaking the von Neumann  bottleneck so that an edge server can serve more users. Specifically, the advantage of batching lies in reusing the part of AI model loaded into a  \emph{graphics processing unit} (GPU) for multiple tasks to avoid frequent memory access \cite{9052125}. As a result, the computation latency per task is reduced and hence the throughput increases \cite{arxiv.2301.12865}. As mentioned, batching should be jointly designed with radio resource allocation to achieve optimal end-to-end performance for multiuser edge AI. Such designs are crucial for 6G AI empowered tactile applications  such as augmented reality (AR) and autonomous driving. In particular, AR requires latency lower than $20$ ms in order to  guarantee an immersive virtual experience for users. However, at its nascent stage, the mentioned area currently has few results \cite{10038543,9843917}. In \cite{10038543}, utilizing the tree-search method, the optimization problem of joint bandwidth allocation and task scheduling to maximize throughput is solved by proposing an efficient tree-search algorithm with intelligent tree pruning. On the other hand, the minimisation of user energy consumption is studied in \cite{9843917} under inference latency constraints. To solve the problem, different algorithms are presented for joint task scheduling and transmission-time control, which allow both online and offline implementation. For simplicity, backlogged tasks and single-batch optimization are assumed in prior work. On one hand, as queuing time is not accounted for, the existing designs cannot provide a guarantee on end-to-end latency between a task arrival and its completion where tasks may find difficulty in supporting real-time applications which require immediate execution of randomly arriving multiuser tasks. On the other hand, techniques from single-batch optimization are inefficient when dealing with the cases with a large number of concurrent tasks with asynchronous  arrivals or with a low arrival rate. In both cases, they can potentially result in long waiting time for those tasks that arrive earlier than others. Optimally forming multiple batches can perform better in such cases but its joint design with radio resource allocation remains as an open problem. 

It is worth mentioning that the issue of asynchronous task arrivals has been addressed in several studies in the MEC area  \cite{9238937,8468240,9292432}. Without targeting a specific task or application, these studies are all based on a generic processor model where computing speeds are measured in, for example, the number of clock cycles required for processing a bit \cite{7879258}. Furthermore, the processor speed is assumed to be controllable by adjusting its clock  frequency that changes its energy consumption following a measurement based model \cite{7572018}. Based on such models, computation-and-radio resources can be jointly managed to maximize the system energy efficiency or throughput under tasks' deadline requirements \cite{7553459}. Due to model abstraction, computing issues as elaborated by the von Neumann and batching have not been studied in the MEC literature. Thereby, the existing solutions are inadequate for solving the current problem of \emph{joint batching and scheduling} (JBAS) for multiuser edge AI. 

In this work, we make an attempt to solve this problem targeting a high-throughput  multiuser edge AI systems under the practical assumptions of asynchronous task arrivals and heterogeneous task deadlines. The problem is challenging for two reasons. First, there are numerous batching related parameters to optimize, namely the number of batches, starting time of individual batches, and the task-batch association. Second, meeting the task deadlines requires the control of end-to-end latency of each task that sums its communication and computation latency. This introduces coupling between batching and scheduling as well as radio resource allocation to scheduled devices. By developing efficient approaches to  solve the complex problem, we develop a framework for optimal JBAS.

The main contributions of this work are summarized as follows.  
\begin{itemize}

\item {\bf Optimal Joint Batching and Scheduling:} The framework of JBAS is designed by solving the JBAS optimization problem. First, we simplify the problem by converting it into an equivalent problem where one variable, the number of batches, is removed. Our technique is to set  the number of batches equal to its maximum by allowing empty batches. Second, the equivalent problem, which is a mixed-integer non-linear program, is made tractable by approximation through the  methods of integer relaxation and $\ell_0$-norm approximation. The resultant convex problem can be efficiently solved using  a proposed algorithm that alternatively solves the following two sub-problems. 
	\begin{itemize}
	\item \emph{Optimization of task-batch association: } The sub-problem is a linear program and its solution can be found using a derived scheme of greedy task selection. The scheme assigns each task to the most suitable batch as measured by a derived metric that accounts for different factors such as batching gain and the task's arrival time and uploading latency. 
	
	\item \emph{Optimization of batch starting time:}  Given the optimal task-batch association, the optimal starting time of each batch is derived in closed form. It is found to be the latest time a batch can start under the deadline and batch causality constraints so that the tasks in the batch use the least radio resources. 
    \end{itemize}
	
	\item {\bf Exploitation of Spectrum Holes:} The combined effects  of synchronized computation duration of tasks in a same batch, their asynchronous arrivals, and heterogeneous uploading durations  create spectrum holes that refer to unused frequency-time resource blocks. We design a spectrum-hole allocation algorithm to optimally exploit spectrum holes to enhance the throughput by admitting  originally unscheduled tasks. The corresponding optimization problem is transformed into a sequence of single-batch sub-problems, each attempting to jointly add new tasks to a specific batch and distribute spectrum holes among them.  The optimal solution for each sub-problem can be found by a linear search that sequentially tests the sub-problem's feasibility given the number of new scheduled tasks.
\end{itemize}
Simulations verify that the proposed JBAS algorithm yields significant performance gains as opposed to  existing  schemes, especially in the scenarios with tight resource constraints. Moreover, the proposed spectrum-hole allocation  scheme is shown to yield significant throughput enhancement. 


The rest of the paper is organized as follows. The system model is described in Section II. The problem of optimal JBAS is formulated in Section III and solved in Section IV. We present the design of the spectrum-hole allocation algorithm in Section V. In Section VI, the extensions to online design with new arriving tasks and frequency-selective channels are discussed. Simulation results are presented in Section VII followed by concluding remarks in Section VIII.

\section{System Model}
\label{sec: models_metrics}

Consider a single-cell system including $K$ devices and an edge server that doubles as an access point, as shown  in Fig.~\ref{fig_systemmodel}. Each device has a single task  that offloads a single data sample (e.g., an image or a video clip)  to the server for inference. 
These tasks  are assumed to share a common pre-trained prediction model, such as a large-scale classifier capable of discerning hundreds of object classes  \cite{5995477}. To reduce communication overhead and protect privacy, each scheduled device uploads a feature vector extracted from raw data using a local model.  
Prior to data uploading, each  device communicates to the server over a control channel the profile of its coming task containing the extracted feature size, arrival instant, and deadline requirement. 
Relevant models and metrics are described in the following sub-sections. 

\begin{figure*}[!t]
	\centering
	\includegraphics[width=0.5\textwidth]{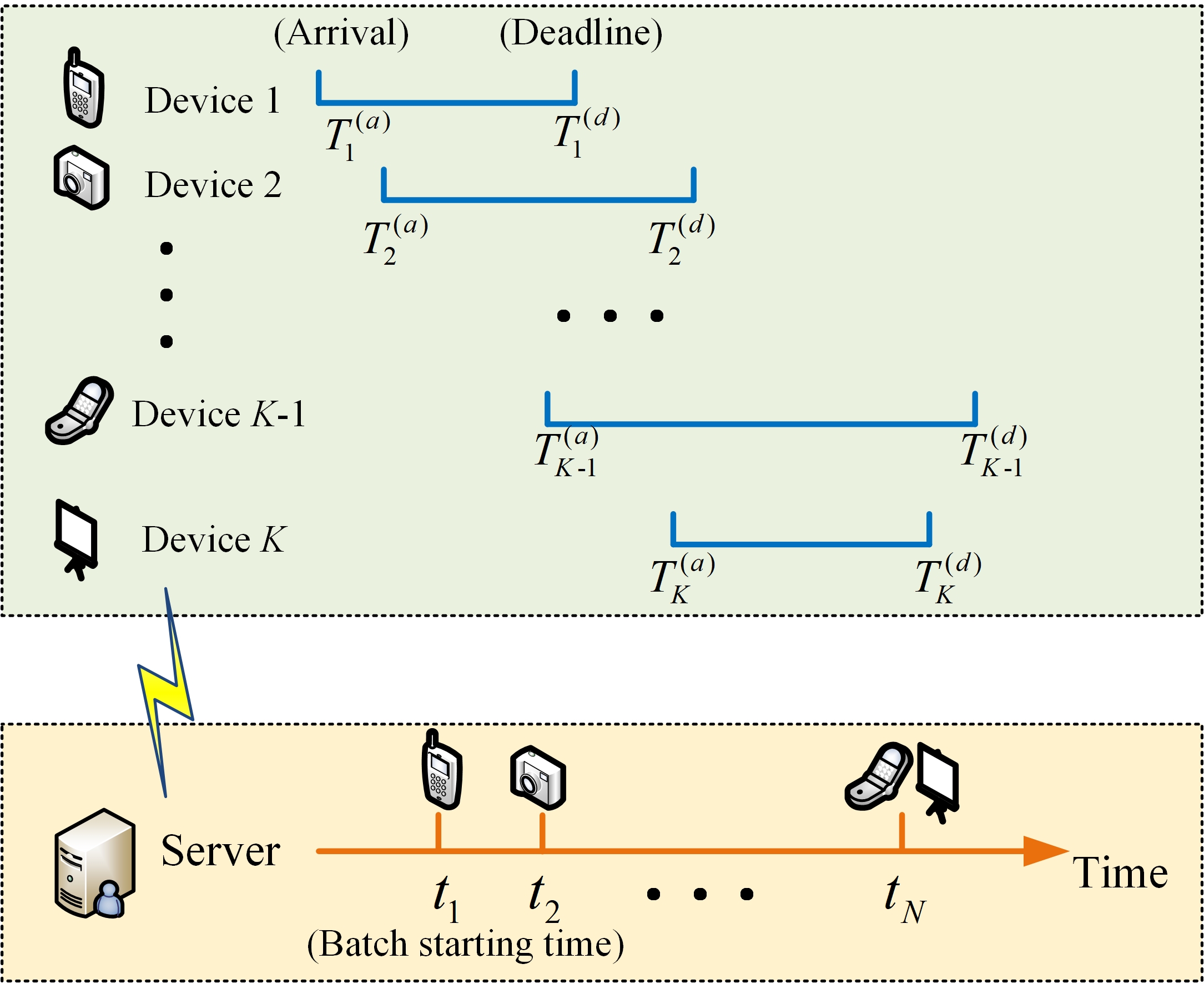}
	\caption{Edge inference system with asynchronous task arrivals.}\label{fig_systemmodel}
	\centering
\end{figure*}


\subsection{Task and Batching Models}
Tasks arrive at devices at random time instants with different sizes and  distinctive end-to-end delay requirements. An arbitrary device, say device $k$, has to finish a task  within the duration of $\left[T^{(a)}_k,T^{(d)}_k\right]$, where $T^{(a)}_k$ and $T^{(d)}_k$ respectively represent the task-arrival time instant and deadline. The  duration consists of three parts: 1) feature uploading phase, 2) task inference phase, and 3) result downloading phase \cite{9606540}. Due to the relatively small size of inference result (e.g., an object label) and high transmit power of the server, the duration of the result downloading phase is assumed negligible. 
To enhance the throughput, the server assembles received tasks  into a number of batches, denoted as $N$, which are fed sequentially to the prediction  model. 
Let $t_n$ with  $n\in\{1,2,\cdots,N\}$ denote the time instant when the processing of the $n$-th batch begins. It follows that  $t_1 < t_2 <\cdots <t_N$. 
To facilitate batching, let $\pi_{k,n}$ represent the association between task $k$ and the $n$-th batch. If the task $k$ is included in the $n$-th batch,  $\pi_{k,n}=1$; otherwise,  $\pi_{k,n}=0$. Since each task should be executed at most once,  
\begin{align} \label{eq_TaskBatchAsso}
	\sum_{n=1}^N \pi_{k,n}\leq1,\quad\forall k\in\mathcal{K}, 
\end{align}
where $\cK$ denotes the set of devices. The server decision on not sewing a device, say device $k$, corresponds to 
$\sum_{n=1}^N \pi_{k,n}= 0$. Then batching reduces to determining the association indicators $\left\{\pi_{k,n}\right\}$. Upon forming batches, the server sequentially inputs batches of feature vectors into the inference model and downloads results as soon as a batch is executed. 

\subsection{Uplink Communication Model}
For simplicity, we consider a frequency non-selective  channel that emerges as propagation distances keep reducing and the extension to frequency-selective channels is provided in Section VI.B. 
Its bandwidth $B$ is divided into $K$ sub-channels that are assigned to the scheduled devices. The bandwidth allocated to device $k$ is denoted as  $B_k$. 
Assume that the channels keep unchanged during the transmission period. The server is assumed to acquire accurate \emph{channel state information} (CSI) useful for resource allocation and device scheduling.   The spectrum efficiency of the channel between device $k$ and the server (in bits/second/Hz) is 
\begin{align} \label{eq_rk}
	r_k=\log_2\left(1+\frac{p_k h_k}{\sigma^2}\right), \quad\forall k\in\cK, 
\end{align} 
where $p_k$ represents the  transmit power, $h_k$ the channel power gain, and $\sigma^2$  the additive white Gaussian noise power. We can write the duration of feature uploading for task/device $k$ as 
\begin{align} \label{eq_tau}
	\tau^o_k = \left(\sum_{n=1}^N\pi_{k,n}t_n\right)-T^{(a)}_k ,  \quad\forall k\in\cK, 
\end{align}
where $T^{(a)}_k$ is the task-arrival instant as defined previously. 
Let $\ell_k$ represent the number of bits in extracted features for task $k$. Then $\ell_k=\tau^o_k B_k r_k$.  From \eqref{eq_tau}, 
\vspace{-0.1em}
\begin{align}
	B_k= \frac{\ell_k}{\left[\left(\sum_{n=1}^N\pi_{k,n}t_n\right)-T^{(a)}_k\right] r_k}, \quad\forall k\in\cK. 
\end{align}

\subsection{Inference Model}
Consider inference with batching \cite{arxiv.2301.12865}. For the $n$-th batch, all the uploaded  feature vectors satisfying $\pi_{k,n}=1$ are assembled and  input as a batch into the server inference model. 
The trained model comprises  multiple  sequential layers. When processing a batch, the server sequentially loads each layer  from the memory and then executes the batch until the batch traverses all layers. 
As found in  the literature, the inference delay increases approximately linearly as the batch size becomes large \cite{arxiv.2301.12865,arxiv.1703.09844}. Given the  association between tasks and batches,  $\left\{\pi_{k,n}\right\}$, the inference delay of the $n$-th batch can be modelled as \cite{arxiv.2301.12865,9843917,10038543}
\begin{align} \label{eq_dn}
	d_n\left(\pi_n\right) =a\pi_n+b, \quad \forall n\in\{1,\cdots,N\},  
\end{align}
where the batch size $\pi_n=\sum_{k=1}^K\pi_{k,n}$ is a positive integer. 
Note that $d_n(\pi_n)$ is a   monotonically increasing function. In the model in \eqref{eq_dn}, $a$ and $b$ depend on the specific inference model \cite{9843917,10038543}.  
Specifically,  $a$ represents the inference delay per task and $b$  the delay of memory access.


\section{Problem Formulation}
\label{sec: problem_formulation}

In this section, the design of JBAS is formulated as an optimization problem with the criterion of maximum system throughput, i.e., the number of completed tasks. According to the inference delay model in \eqref{eq_dn}, increasing the batch size can  reduce the inference delay per task. However, due to heterogeneous task arrival instants and deadlines, waiting for more tasks to arrive to form a batch hinders the completion of those with early deadlines. On the other hand, to start a batch earlier requires more radio  resources so as to finish uploading the associated tasks in time. As a result, there exist two tradeoffs: one between the  batch size and batch starting  instant and the other between communication  and computation resources.  Furthermore, the association between tasks and batches also needs to be optimized. 

Several practical constraints are considered. 
The first is the task-causality constraint, namely that the processing of a batch cannot begin until the arrivals of all associated tasks: 
\begin{align} \label{eq_Tka}
\pi_{k,n} T^{(a)}_k< t_n,\quad\forall k\in\cK,\forall n \in\{1,\cdots,N\}. 
\end{align}
The second constraint enforces the deadline requirements  of scheduled tasks: 
\begin{align} \label{eq_Tkd}  \pi_{k,n}\left[t_n + d_n\left(\pi_n\right)\right] \leq T^{(d)}_k, \quad\forall k\in\cK,\forall n \in\{1,\cdots,N\}. 
\end{align}
Note that  when a task, say task  $k$,  is not associated with the $n$-th batch, i.e., $\pi_{k,n}=0$,  constraints \eqref{eq_Tka} and \eqref{eq_Tkd} are always satisfied.
The third constraint reflects sequential batch processing, namely that the $\left(n+1\right)$-th batch is not processed   until the $n$-th batch  finishes its inference: 
\begin{align}
 t_n + d_n\left(\pi_n\right) \leq t_{n+1}, \quad\forall n \in\{1,\cdots,N-1\}. 
\end{align}
Last, the bandwidth constraint is given as 
\begin{align}
\sum_{k=1}^K\sum_{n=1}^N\pi_{k,n}\frac{\ell_k}{r_k\tau_k^o}\leq B. 
\end{align}
 
Under the above constraints, we aim at optimizing the bandwidth allocation, the number of batches,  their starting instants, as well as the task-batch association. Note that a task that is not assigned to any batch is not scheduled for execution. Then the JBAS optimization problem is formulated as 
\begin{equation*}  \text{(P1)} \quad\ 
	\begin{aligned}
		\max\limits_{ \{t_n\},\{\pi_{k,n}\},N}\quad & \sum_{k=1}^K\sum_{n=1}^N \pi_{k,n},  \\
		\mathrm{s.t. }\quad\ \quad 
		& \pi_{k,n} T^{(a)}_k < t_n,\quad\forall k\in\cK,\forall n \in\{1,\cdots,N\}, \\
		& \pi_{k,n}\left[t_n + d_n\left(\pi_n\right)\right] \leq T^{(d)}_k, \quad\forall k\in\cK,\forall n \in\{1,\cdots,N\}, \\
		& t_n + d_n\left(\pi_n\right) \leq t_{n+1}, \quad\forall n \in\{1,\cdots,N-1\}, \\
		& \sum_{k=1}^K\sum_{n=1}^N \pi_{k,n}\frac{\ell_k}{r_k\tau_k^o} \leq B, \\
		&\sum_{n=1}^N \pi_{k,n}\leq 1,\quad\forall k\in\mathcal{K},\\
		&\pi_{k,n}\in\{0,1\}, \quad\forall k\in\cK,\forall n \in\{1,\cdots,N\}, \\
		& N\in\mathbb{Z}^+, N\leq K.
	\end{aligned}
\end{equation*}
Problem (P1) is non-convex and NP-hard to solve due to the binary task-batch association indicators as well as the  coupling  between optimization variables \cite{10038543,KORTE1981}. Furthermore, the variable number of batches, $N$,  can change the cardinalities of batch starting instants, $\{t_n\}$, as well as the association indicators,  $\{\pi_{k,n}\}$, further complicating this problem. 

\section{Optimal JBAS Algorithm}
In this section, we design an efficient algorithm for JBAS by approximately solving Problem (P1). 
The proposed solution approach is to transform the problem to an equivalent, simpler one with the number of bathes fixed. Then applying the method of integer relaxation allows the equivalent problem to be solved using an alternating optimization algorithm. Its complexity is analyzed. 



\subsection{A Tractable Solution Approach} 
\subsubsection{An Equivalent Problem}
First, Problem (P1) can be transformed into the following equivalent problem:  
\vspace{-1em}
\begin{equation*}\text{(P2)}
\begin{aligned}
\max\limits_{\{t_n\},\{\pi_{k,n}\},N}\quad & \sum_{k=1}^K\sum_{n=1}^N \pi_{k,n},  \\
   \mathrm{s.t. }\quad\ \quad 
   & \pi_{k,n} T^{(a)}_k <  t_n,\quad\forall k\in\cK,\forall n \in\{1,\cdots,N\}, \\
   & t_n + d_n\left(\pi_n\right) \leq T^{(d)}_k + \left(1-\pi_{k,n}\right)\Xi, \quad\forall k\in\cK,\forall n \in\{1,\cdots,N\}, \\
   & t_n + d_n\left(\pi_n\right) \leq t_{n+1}, \quad\forall n \in\{1,\cdots,N-1\}, \\
   & \sum_{k=1}^K\sum_{n=1}^N \pi_{k,n}\frac{\ell_k}{r_k\left(t_n- T_k^{(a)}\right)} \leq B, \\
   &\sum_{n=1}^N \pi_{k,n}\leq 1,\quad\forall k\in\mathcal{K},\\
   &\pi_{k,n}\in\{0,1\}, \quad\forall k\in\cK,\forall n \in\{1,\cdots,N\}, \\
   & N\in\mathbb{Z}^+, N\leq K.
\end{aligned}
\end{equation*}
where the constant $\Xi\triangleq \max_{k\in\mathcal{K}}T^{(d)}_k + d_N\left(K\right)$. 
Problem (P2) is different from the conventional \emph{mixed integer nonlinear programming} (MINLP) problem since the number of variables varies with the number of batches, $N$. Without loss of generality, we  propose to mend the difference by fixing $N$ as $N=K$ by allowing the existence of empty batches. This results in the following MINLP problem: 
\vspace{-0.1em}
\begin{equation*}\text{(P3)} 
	\begin{aligned}
		\quad\quad \max\limits_{\{t_n\},\{\pi_{k,n}\}}\quad & \sum_{k=1}^K\sum_{n=1}^N \pi_{k,n},  \\
		\mathrm{s.t. }\quad\ \quad 
		& \pi_{k,n} T^{(a)}_k < t_n,\quad\forall k\in\cK,\forall n \in\{1,\cdots,N\}, \\
		& t_n + d_n\left(\pi_n\right) \leq T^{(d)}_k + \left(1-\pi_{k,n}\right)\Xi, \quad\forall k\in\cK,\forall n \in\{1,\cdots,N\}, \\
		& t_n + d_n\left(\pi_n\right) \leq t_{n+1}, \quad\forall n \in\{1,\cdots,N-1\}, \\
		& \sum_{k=1}^K\sum_{n=1}^N \pi_{k,n}\frac{\ell_k}{r_k\left(t_n- T_k^{(a)}\right)} \leq B, 
  \end{aligned}
  \end{equation*}
  \begin{equation*}
      \begin{aligned}
		&\sum_{n=1}^N \pi_{k,n}\leq 1,\quad\forall k\in\mathcal{K},\\
		&\pi_{k,n}\in\{0,1\}, \quad\forall k\in\cK,\forall n \in\{1,\cdots,N\}, 
	\end{aligned}
\end{equation*}
where the corresponding inference delay evolves as: 
\begin{align} \label{eq_dn2}
 d_n\left(\pi_n\right)=\left\{\begin{aligned}
		&a\pi_{n}+b, \quad\text{if }\pi_{n}>0,\\
		&0, \ \ \quad\quad\quad\text{if }\pi_{n}=0,
	\end{aligned}\right.
\end{align}
for all $n$. The following theorem gives the equivalence between Problems (P2) and (P3). 
\begin{theorem} \label{theorem_3}
	Problems (P3) and (P2) are equivalent in the sense that their optimal objectives are identical.
\end{theorem}
\noindent The proof is provided in Appendix~\ref{proof_theorem_3}.   $\hfill\square$

Theorem~\ref{theorem_3} allows us to solve Problem (P2) by solving Problem (P3) that leverages MINLP. 

\subsubsection{Integer Relaxation and Alternating Optimization}
To solve (P3), the method of integer  relaxation is adopted to obtain an   approximate solution (see e.g., \cite{9179773}). Specifically, the binary variables $\{\pi_{k,n}\}$ are relaxed as continuous ones belonging to $[0,1]$. It should be emphasized that the relaxation does not compromise the optimaliy as discussed in Remark~\ref{Re:Optimality}. 
Due to the existence of empty batches, the inference delay function $d_n\left(\pi_n\right)$ has a step at $\pi_n=0$, making Problem (P3) nonconvex. To address the issue,  it can be rewritten in a form comprising $\ell_0$-norm as 
\begin{align} \label{eq_D}
 d_n\left(\pi_n\right)=a\pi_{n}+b\mathds{1}_{\left\{\pi_{n}\right\}}=a\pi_{n}+b\left\|\pi_{n}\right\|_0, 
\end{align}
where $\|\cdot\|_0$ is $\ell_0$-norm and  $\mathds{1}_{\left\{x\right\}}$ is the indicator function that is $1$ if $x>0$ and $0$ otherwise. 
The non-smooth $\ell_0$-norm can be well approximated by a series of convex weighted $\ell_1$-norms, which is a  commonly used technique in compressive sensing (see e.g., \cite{7437385,7102696}). Using this technique, the $\ell_0$-norm term in \eqref{eq_D} can be approximated by an asymptotically equivalent term as 
\begin{align} 
\left\|\sum_{k=1}^K\pi_{k,n}\right\|_0=\lim_{\delta\rightarrow0}\frac{\ln\left(1+\delta^{-1}\sum_{k=1}^K\pi_{k,n}\right)}{\ln\left(1+\delta^{-1}\right)}. 
\end{align}
Since the logarithmic function is concave and upper bounded by the first-order term of Taylor's expansion, we have
\vspace{-0.1em}
\begin{align}\label{eq_ell0}
\left\|\sum_{k=1}^K\pi_{k,n}\right\|_0\leq&\theta_n^{(r)}\sum_{k=1}^K\pi_{k,n}+\psi_{n}^{(r)}, 
\end{align} 
with 
\begin{align} \label{eq_thetan}
\theta_n^{(r)}=\frac{\delta^{-1}\left(1+\delta^{-1}\sum_{k=1}^K\pi_{k,n}^{(r)}\right)^{-1}}{\ln\left(1+\delta^{-1}\right)},
\end{align}
 and
 \begin{align} \label{eq_psin}
\psi_n^{(r)}=\frac{\ln\left(1+\delta^{-1}\sum_{k=1}^K\pi_{k,n}^{(r)}\right)+\left(1+\delta^{-1}\sum_{k=1}^K\pi_{k,n}^{(r)}\right)^{-1}-1}{\ln\left(1+\delta^{-1}\right)},
 \end{align} 
where $\pi_{k,n}^{(r)}$ represents  the  value of $\pi_{k,n}$ at the previous iteration and $\delta$ is a sufficiently small constant. The equality in \eqref{eq_ell0} holds if and only if $\pi_{k,n}=\pi_{k,n}^{(r)}$ for all $(n,k)$. Through the above iterative updates of $\theta_n^{(r)}$ and $\psi_n^{(r)}$,   the difference between $\left\|\pi_{n}\right\|_0$ and its first-order term of Taylor's expansion diminishes until the equality in \eqref{eq_ell0} holds.  
Then substituting \eqref{eq_ell0} into \eqref{eq_D}, the inference delay function can be approximated as  
\begin{align} \label{eq_D1}
d_n\left(\pi_n\right)\approx\left(a+b\theta_n^{(r)}\right)\sum_{k=1}^K\pi_{k,n}+b\psi_n^{(r)},  \quad \forall n, 
\end{align}
which is continuous and linear. Using \eqref{eq_D1},  Problem (P3) can be approximated as 
\vspace{-0.5mm}
\begin{equation*}\text{(P4)} 
	\begin{aligned}
		\max\limits_{\{t_n\},\{\pi_{k,n}\}}\quad & \sum_{k=1}^K\sum_{n=1}^N \pi_{k,n},  \\
		\mathrm{s.t. }\quad\ \quad 
		& \pi_{k,n} T^{(a)}_k < t_n,\quad\forall k\in\cK,\forall n \in\{1,\cdots,N\}, \\
		& t_n + \left(a+b\theta_n^{(r)}\right)\sum_{k=1}^K\pi_{k,n}+b\psi_n^{(r)} \leq T^{(d)}_k + \left(1-\pi_{k,n}\right)\Xi, \forall k\in\cK,\forall n \in\{1,\cdots,N\}, \\
		& t_n + \left(a+b\theta_n^{(r)}\right)\sum_{k=1}^K\pi_{k,n}+b\psi_n^{(r)} \leq t_{n+1}, \quad\forall n \in\{1,\cdots,N-1\}, \\
		& \sum_{k=1}^K\sum_{n=1}^N \pi_{k,n}\frac{\ell_k}{r_k\left(t_n- T_k^{(a)}\right)} \leq B, \\
		&\sum_{n=1}^N \pi_{k,n}\leq 1,\quad\forall k\in\mathcal{K},\\
		&0\leq \pi_{k,n} \leq 1, \quad\forall k\in\mathcal{K},\forall n \in\{1,\cdots,N\}.
	\end{aligned}
\end{equation*}
This problem is convex and can be readily solved utilizing the approach of alternating optimization.  We propose to alternate solving two reduced-dimension sub-problems as described in the following sub-sections. As a result, the complexity is dramatically reduced as opposed to directly solving Problem (P4) and furthermore useful insight can be obtained. It is worth mentioning that alternating optimization provides no guarantee on reaching the global optimal point since  the constraints in Problem (P4) are not \emph{box constraints} (see, e.g., \cite{grippo2000convergence}).  The complete algorithm is presented in Algorithm~\ref{alg1}.

\subsection{Optimal Task-Batch  Association}
The first sub-problem results from fixing the starting time of batches, $\left\{t_n\right\}$, in Problem (P4). Then it reduces to a linear program. The dual problem of (P4) with respect to task-batch association,  $\{\pi_{k,n}\}$, is given as 
\begin{align} \label{eq_dual}
	\min_{\{\beta_{k,n}\},\{\gamma_{k,n}\},\rho} G\left(\beta_{k,n},\gamma_{k,n},\rho\right),
\end{align}
where $G\left(\beta_{k,n},\gamma_{k,n},\rho\right)$ is the dual function that solves  
\begin{align} \label{eq_G}
		\max_{\{\pi_{k,n}\}}\ \   &\mathcal{L}\left(\pi_{k,n},\beta_{k,n},\gamma_{k,n},\rho\right),\\
		\mathrm{s.t. } \ \  & \sum_{n=1}^N\pi_{k,n}\leq 1,\quad\forall k\in\cK,\nonumber\\
		& 0\leq\pi_{k,n}\leq1,\quad\forall k\in\cK,\forall n\in\{1,\cdots,N\}.\nonumber
\end{align}
In \eqref{eq_G}, $\mathcal{L}\left(\pi_{k,n},\beta_{k,n},\gamma_{k,n},\rho\right)$ denotes the partial Lagrangian function of Problem (P4): 
\vspace{-0.5mm}
\begin{align} \label{eq_Lagrange}
	&\mathcal{L}\left(\pi_{k,n},\alpha_{k,n},\beta_{k,n},\gamma_{k,n},\rho\right)=\sum_{k=1}^K\sum_{n=1}^N \pi_{k,n} \nonumber\\
	&-\sum_{k=1}^K\sum_{n=1}^N\beta_{k,n}\left[t_n + \left(a+b\theta_n^{(r)}\right)\left(\sum_{k=1}^K \pi_{k,n}\right)+b\psi_n^{(r)} - T^{(d)}_k - \left(1-\pi_{k,n}\right)\Xi\right]\nonumber\\
	&-\sum_{n=1}^{N-1}\gamma_{n}\left[t_n + \left(a+b\theta_n^{(r)}\right)\left(\sum_{k=1}^K \pi_{k,n}\right)+b\psi_n^{(r)} - t_{n+1}\right] -\rho\left(\sum_{k=1}^K\sum_{n=1}^N \pi_{k,n}\frac{\ell_k}{r_k\left(t_n- T_k^{(a)}\right)} - B\right)\nonumber\\
\end{align} 
where $\beta_{k,n}$, $\gamma_{k,n}$, and $\rho$ are non-negative Lagrange multipliers associated with the deadline, batch causality, and bandwidth allocation constraints, respectively. Besides, we let $\gamma_{0}=\gamma_{N}=0$ for consistency. We can observe that  \eqref{eq_Lagrange} is linear with respect to $\pi_{k,n}$. Therefore, to maximize the Lagrange function with fixed multipliers, the optimal $\pi_{k,n}$ is either zero or one.  Specifically, for all $n$, if $\pi_{k,n}$ are less than or equal to zero, task $k$ is not scheduled, i.e., $\pi_{k,n}=0$; otherwise, this task is associated with the batch that has the largest coefficient: 
\vspace{-1mm}
\begin{align} \label{eq_pikn}
	\pi_{k,n}^*=\left\{\begin{aligned}
		&1, \quad \text{if }n=\arg\max_{n\in\{1,\cdots,N\}} \mu_{k,n},\\
		&0, \quad\text{otherwise},
	\end{aligned}\right. 
\end{align}
where $\mu_{k,n}=1 - \left(a+b\theta_n^{(r)}\right)\sum_{k=1}^K\beta_{k,n}-\left(a+b\theta_n^{(r)}\right)\gamma_{n} - \Xi\beta_{k,n} -\rho\frac{\ell_k}{r_k\left(t_n- T_k^{(a)}\right)}$.  If there are multiple batches  satisfying $\arg\max_{n\in\{1,\cdots,N\}} \mu_{k,n}$, we can choose any of them due to the non-strict convexity of Problem (P3). Then substituting \eqref{eq_pikn} into \eqref{eq_G}, we can obtain $G\left(\beta_{k,n},\gamma_{k,n},\rho\right)$. 

\begin{remark} 

\emph{(Favourable Task Conditions)
According to \eqref{eq_pikn}, one can infer that for a task, as the channel condition becomes worse, its likelihood of being scheduled reduces as uploading the task requires more radio resources or else incurs higher latency. 
Moreover, early task-arrival time increases the probability that a task is successfully executed due to the following two reasons: 
1) the larger batching gain, 
and 2) the longer communication time that increases the probability of successful feature uploading. 
Last, by combining \eqref{eq_thetan} and \eqref{eq_pikn}, we can observe  that a task prefers a larger batch as its inference delay per task is smaller due to the batching gain.}
\end{remark}

Given  $G\left(\beta_{k,n},\gamma_{k,n},\rho\right)$, we attempt to solve the dual problem \eqref{eq_dual} to get the the optimal dual variables $\{\beta_{k,n}\},\{\gamma_{k,n}\},\rho$. Note that $G\left(\beta_{k,n},\gamma_{k,n},\rho\right)$ is not differentiable in general due to the discontinuous selection operations in obtaining the optimal $\pi_{k,n}^*$. To this end, the value of dual variables is updated by the sub-gradient method\cite{cot}. Thus, through iteratively optimizing primal variables and dual variables, the optimal tasks and batches association $\pi_{k,n}$ with fixed $t_n$ can be obtained directly without rounding according to the following remark.  

\begin{remark}\label{Re:Optimality}
\emph{(Optimality of Task-Batch Association) 
We can observe that for task $k$, there exists at most a single element among $\{\pi_{k,n}\}$ that is equal to one while others are set as zero according to \eqref{eq_pikn}. This indicates that although the feasible range of $\pi_{k,n}$ is relaxed to be continuous, the optimal solution to
Problem (P4) with respect to $\pi_{k,n}$ always satisfies the binary constraint $\pi_{k,n}\in\{0,1\}$ for all $(n,k)$. Hence, the relaxation of $\pi_{k,n}$ does not compromise the optimality of the original Problem (P3). }
\end{remark}



\vspace{-1cm}
\subsection{Optimal Batch Starting Time}
The other sub-problem results from fixing the task-batch association, $\{\pi_{k,n}\}$, in Problem (P4).  As a result, the sub-problem is written as 
\vspace{-0.1em} 
\begin{equation*}\text{(P5)}
	\begin{aligned}
		\max\limits_{\{t_n\}}\quad & \sum_{k=1}^K\sum_{n=1}^N \pi_{k,n},  \\
		\mathrm{s.t. }\quad 
		& \pi_{k,n} T^{(a)}_k < t_n,\quad\forall k\in\cK,\forall n \in\{1,\cdots,N\}, \\
		& t_n + d_n\left(\pi_n\right) \leq T^{(d)}_k + \left(1-\pi_{k,n}\right)\Xi, \quad\forall k\in\cK,\forall n \in\{1,\cdots,N\}, \\
		& t_n + d_n\left(\pi_n\right) \leq t_{n+1}, \quad\forall n \in\{1,\cdots,N-1\}, \\
		& \sum_{k=1}^K\sum_{n=1}^N \pi_{k,n}\frac{\ell_k}{r_k\left(t_n- T_k^{(a)}\right)} \leq B. 
	\end{aligned}
\end{equation*}
\begin{theorem} \label{theorem_2}
The optimal starting time of the $n$-th batch, which solves Problem (P5), is given as 
	\begin{align} \label{eq_tn}
		t^*_n=\left\{\begin{aligned}
			&\min\left\{\chi^{(d)}_{n},t_{n+1}\right\} - d_n\left(\pi_n\right),\quad\text{if }\pi_{n}>0,\\
			&t^*_{n+1},\quad\quad\quad\quad\quad\quad\quad\quad\quad\quad\ \text{otherwise}, 
		\end{aligned}\right. \quad\left(n=N,\cdots,1\right),
	\end{align}
	where $\chi^{(d)}_n=\min_{ k\in\cK_n}T^{(d)}_k$ denotes the minimum deadline among all the tasks processed in the $n$-th batch, i.e., $\cK_n=\left\{k|\pi_{k,n}=1\right\}$ $\left(\forall n\in\{1,\cdots,N\}\right)$, and $t^*_{N+1}=\Xi$.  
\end{theorem}
\noindent The proof is provided in Appendix~\ref{proof_theorem_2}.  $\hfill\blacksquare$

From Theorem~\ref{theorem_2}, we can observe that with fixed $\pi_{k,n}$, $t_n$ is only determined by the deadlines of tasks in the $n$-th batch and $t_{n+1}$.  The starting time of a batch, $t_n$, is set as the latest starting time that can ensure the latency and batch causality  constraint such that the scheduled tasks occupy the least radio resources.  

\subsection{Complexity Analysis}

The complexity of Algorithm~\ref{alg1} is largely attributed to solving the two  subproblems solved in the preceding subsections. The complexity in optimizing the  task-batch association is  $\mathcal{O}\left(K^2/\sqrt{\epsilon}\right)$ based on \eqref{eq_pikn}, where $\epsilon$ represents the predefined accuracy of the dual method \cite{boyd2004convex}. The complexity of calculating the optimal batch starting instants is  $\mathcal{O}\left(K\right)$ according to \eqref{eq_tn}. The overall complexity is given by $\mathcal{O}\left(L\left(K^2/\sqrt{\epsilon}+K\right)\right)$, where $L$ denotes the average number of iterations in Algorithm~\ref{alg1}.  

Last, the complexity of  Algorithm~\ref{alg1} is much lower than the conventional interior point method for directly solving Problem (P4) whose complexity is $\mathcal{O}\left(L\left(K^2+K\right)^{3.5}\right)$ \cite{boyd2004convex}.

\begin{algorithm}[t]
	\begin{small}
		\caption{JBAS Algorithm}
		\label{alg1}
		
		Initialize $t_n=\frac{\max_{k\in\cK}T^{(d)}_k-\min_{k\in\cK}T^{(a)}_k}{N-1}\times (n-1)+\min_{k\in\cK}T^{(a)}_k$ $\left(\forall n \in\{1,\cdots,N\}\right)$, $t_{N+1}=\Xi$, $\pi_{k,n}^{(r)}=0$ $\left(\forall k\in \cK,\forall n\in\{1,\cdots,N\}\right)$ and required precision.  
		
		\Repeat{\textnormal{the objective of Problem  (P3) converges}}{ 
			Initialize $\{\beta_{k,n}\}$, $\{\gamma_n\}$, $\rho$. 
			
			\Repeat{\textnormal{the objective of problem  (9) converges}}{

				Obtain the association between tasks and batches $\{\pi_{k,n}\}$ according to  \eqref{eq_pikn}. 
				
				Update dual vairables  $\{\beta_{k,n}\}$, $\{\gamma_n\}$, $\rho$ using the sub-gradient method.
			}
			
			\For{$n=N,\cdots,1$}{Obtain the startup time of batches $\{t_n\}$ according to Theorem~\ref{theorem_2}.}
		Update $\theta_n^{(r)}$ and $\psi_n^{(r)}$ according to \eqref{eq_thetan} and \eqref{eq_psin}, respectively. 
		} 
		Output  the optimal $\{\pi_{k,n}\}$ and $\{t_n\}$. 
	\end{small}  
\end{algorithm}

\begin{figure*}[!t]
	\centering
	\includegraphics[width=0.7\textwidth]{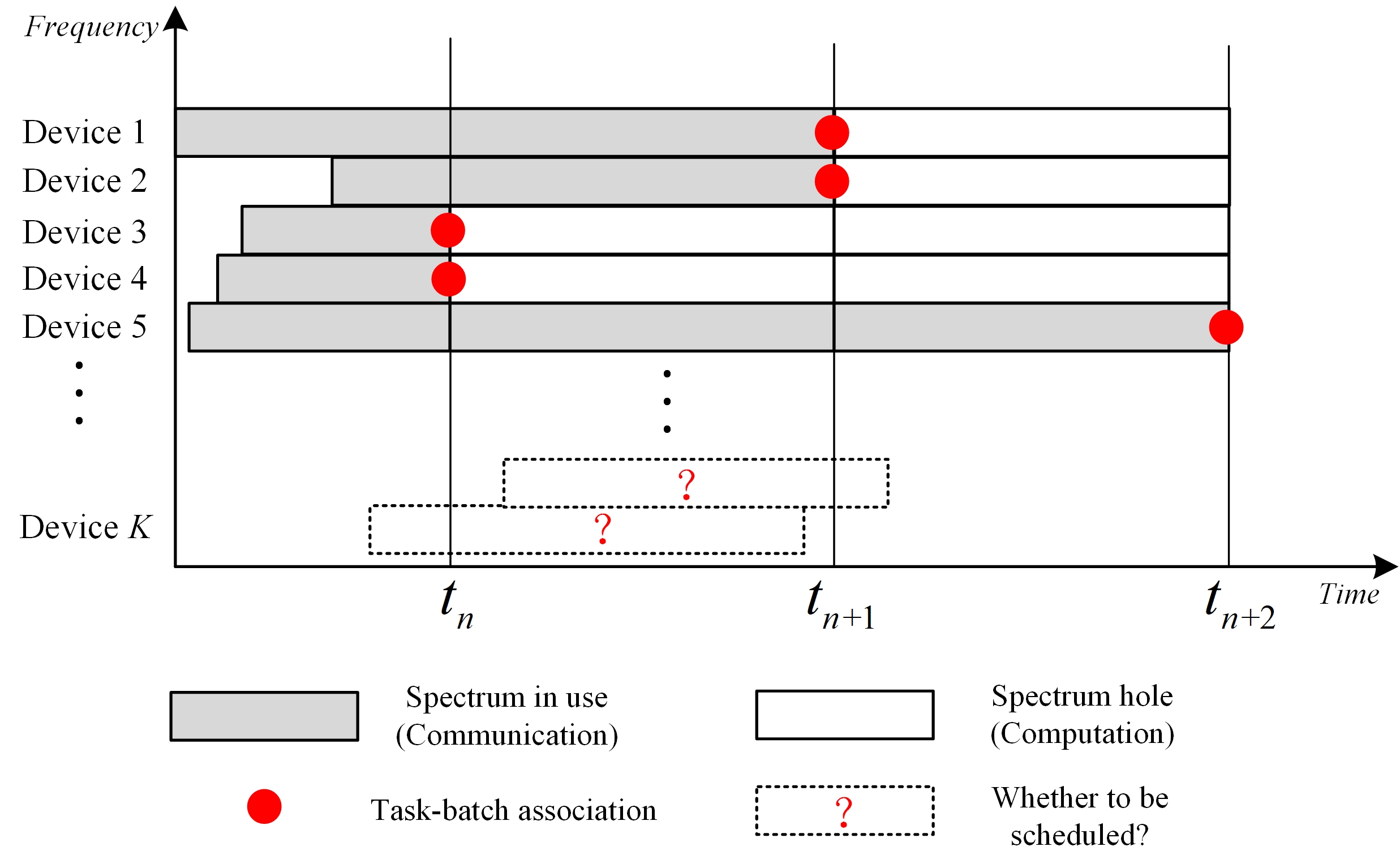}
    \vspace{-0.5cm}
	\caption{Illustration of spectrum-holes.}\label{fig_reallocation}
	\centering
 \vspace{-0.7cm}
\end{figure*}

\vspace{-0.5cm}
\section{Exploiting Spectrum Holes} 
In Section II, individual devices are assigned dedicated frequency bands for identical uploading durations to allow tractable design. Nevertheless, due to  heterogeneous task arrivals and communication latency, there exist \emph{spectrum holes} that can be exploited to further increase the throughput.  As defined, a spectrum hole refers to an unused frequency-time resource block. 
In this section, an algorithm for spectrum hole allocation algorithm is designed by formulating and solving a corresponding throughput maximization problem. 

\vspace{-0.5cm}
\subsection{Spectrum-Hole Allocation Problem} 
As shown in Fig.~\ref{fig_reallocation},  upon arriving time instant $t_n$ for all $n$, the tasks inferred in the $n$-th batch finish their transmission, and thus spectrum holes (i.e., the bandwidth left by scheduled tasks) can be allocated to the unscheduled tasks to improve the throughput or the scheduled tasks to reduce the latency. It should be noted that fixing $\left\{t_n\right\}$ as computed using Algorithm~\ref{alg1}, makes it difficult to insert additional tasks for inference due to the tight deadlines according to \eqref{eq_tn}. This means that even those unscheduled tasks can be uploaded to the server exploiting spectrum holes, they cannot be executed in the original batches without interrupting originally tasks.  
The challenge faced in the current problem lies in adjusting $\left\{t_n\right\}$ to accommodate new tasks without causing the failure of any existing task to meet its deadline. 
Denote $\cF$ as the set of unscheduled tasks based on  Algorithm~\ref{alg1}. For each $t_n$, let $\cS_n$ denote the set of tasks associated with the $n$-th batch.  Hence, at each arrival time, say $t_n$, the total bandwidth of spectrum holes, denoted by $\bar B_n=\sum_{i=1}^{n}\sum_{k\in\cS_i}B_k$, can be allocated to unscheduled tasks in $\cF$ so as to improve the throughput. Moreover, we have to adjust starting time of the $\left(n+1\right)$-th batch such that new scheduled  tasks can be inserted into current batch without causing any  original scheduled tasks to miss its deadline. 
Mathematically, at each checkpoint $t_n$ with $n\in\{1,\cdots,N-1\}$, 
we let $\tilde \ell_k$ represent the data size to be transmitted for task $k$ in the duration from $t_n$ to $t_{n+1}$. Specifically, for tasks in $\cS_{n+1}$,  $\tilde \ell_k$ is given by $B_kr_k\left(\bar t_{n+1}-\max\left\{ T_k^{(a)},\bar t_n\right\}\right)$, and for tasks in $\cF$, $\tilde \ell_k$ is equivalent to $\ell_k$. 
Optimization variables $\bar t_{n+1}$ and $\bar\cK_{n+1}$ represent the adjusted startup instant of the $\left(n+1\right)$-th batch and new-scheduled tasks in the $(n+1)$-th batch, respectively. 
Besides, we let $\bar t_1=t_1$. 
Then, we solve the following optimization problem: 
 \begin{equation*}\text{(P6)}
 	\begin{aligned}
 		\max_{\bar t_{n+1},\bar\cK_{n+1}}\quad & |\bar\cK_{n+1}|,  \\
 		\mathrm{s.t. }\quad\ \  
 		& \max\left\{\max_{k\in\bar\cK_{n+1}\cup\cS_{n+1}} T^{(a)}_k, \bar t_n+d(|\cS_{n}|)\right\} < \bar t_{n+1},\\
 		&\bar t_{n+1} + d\left(|\bar\cK_{n+1}\cup\cS_{n+1}|\right) \leq \min\left\{\min_{k\in\bar\cK_{n+1}\cup\cS_{n+1}} T^{(d)}_k,t_{n+2}\right\},  \\
 		& \sum_{k\in\bar\cK_{n+1}\cup\cS_{n+1}} \frac{ \tilde \ell_k}{r_k\left(\bar t_{n+1}-\max\left\{ T^{(a)}_k,\bar t_n\right\}\right)} \leq \sum_{i=1}^{n+1}\sum_{k\in\cS_i}B_k, \\
 		&\bar\cK_{n+1}\subseteq\cF. 
 	\end{aligned}
 \end{equation*}
In (P6), the first and second constraints specify the causality of the new-scheduled tasks $\bar\cK_{n+1}$ and  previously determined to be scheduled tasks $\cS_{n+1}$ for the $(n+1)$-th batch. 
The third constraint implies that the allocated bandwidth to tasks in $\bar\cK_{n+1}$ and $\cS_{n+1}$ should not exceed the sum of spare bandwidth and previously determined bandwidth for tasks in $\cS_{n+1}$, which ensures that bandwidth reallocation does not influence the bandwidth allocated to other  tasks being transmitted. 
To this end, the spectrum-hole allocation  problem is formulated as a sequence of subproblems, each corresponding to a one-batch optimization problem of tasks scheduling and bandwidth allocation for the $(n+1)$-th batch $(\forall  n\in\{1,\cdots,N-1\})$ at checkpoint $t_n$. By solving Problem (P6),  we attempt to increase the number of scheduled tasks and decrease the total delay simultaneously.

\subsection{Solution Approach}
Problem (P6) is difficult to solve due to the combinatorial nature of this problem. In the following, we show that the optimal solution of (P6) can be obtained by solving a series of feasibility problems each corresponding to a fixed number of new-scheduled tasks $|\bar\cK_{n+1}|=\Pi\in\left\{0,1,\cdots,|\cF|\right\}$: 
\begin{equation*}\text{(P7)}
	\begin{aligned}
		\mathrm{find} \quad & \bar t_{n+1},\bar\cK_{n+1},  \\
		\mathrm{s.t. }\ \ \  
		& \max\left\{\max_{k\in\bar\cK_{n+1}\cup\cS_{n+1}} T^{(a)}_k, \bar t_n+d(|\cS_{n}|)\right\} < \bar t_{n+1},\\
		&\bar t_{n+1} + d\left(\Pi+|\cS_{n+1}|\right) \leq \min\left\{\min_{k\in\bar\cK_{n+1}\cup\cS_{n+1}} T^{(d)}_k,t_{n+2}\right\},  \\
		& \sum_{k\in\bar\cK_{n+1}\cup\cS_{n+1}} \frac{ \tilde \ell_k}{r_k\left(\bar t_{n+1}-\max\left\{ T^{(a)}_k,\bar t_n\right\}\right)} \leq \sum_{i=1}^{n+1}\sum_{k\in\cS_i}B_k, \\
		&\bar\cK_{n+1}\subseteq\cF, |\bar\cK_{n+1}|=\Pi. 
	\end{aligned}
\end{equation*}
\begin{proposition} \label{proposition_1}
	Denote the optimal solution of Problem (P6) by  $\bar\cK_{n+1}^*$, $\bar{t}_{n+1}^*$ and let  $|\bar\cK_{n+1}^*|=\Pi^*$. 
    Problem (P7) is feasible if and only if $\Pi<\Pi^*$. 
\end{proposition}
\noindent The proof is provided in Appendix~\ref{proof_proposition_1}. $\hfill\square$

Proposition~\ref{proposition_1} implies that we can adopt the  bisection method to solve Problem (P6) as described below. Denote the upper bound and lower bound of $\Pi$ by  $\Pi_{\text{up}}$ and $\Pi_{\text{low}}$, respectively. Let $\Pi=\lfloor\left(\Pi_{\text{up}}+\Pi_{\text{low}}\right)/2\rfloor$, where $\lfloor\cdot\rfloor$ denotes the round down operation,  and solve Problem (P7). If (P7) is feasible, which means that the optimal $\Pi^*$ is no smaller than $\Pi$, we set $\Pi_{\text{low}}=\Pi$. Otherwise,  the optimal $\Pi^*$ is no larger than $\Pi$, thus setting $\Pi_{\text{up}}=\Pi$. 
Repeating these procedures until $\Pi_{\text{up}}-\Pi_{\text{low}}\leq1$. Then if Problem (P7) is feasible when  $\Pi=\Pi_{\text{up}}$, the optimal $\Pi^*=\Pi_{\text{up}}$; otherwise, we have $\Pi^*=\Pi_{\text{low}}$. 

However, it remains to solve Problem (P7) with given $\Pi$. Note that with given $\min_{k\in\bar\cK_{n+1}}T^{(d)}_{k}$, the optimal $\left(n+1\right)$-th batch starting time $\bar{t}^*_{n+1}$ should be given as $ \min\left\{\min_{k\in\bar\cK_{n+1}\cup\cS_{n+1}} T^{(d)}_k,t_{n+2}\right\}$ $-d\left(\Pi+|\cS_{n+1}|\right)$ since the required bandwidth for tasks decreases with $\bar{t}_{n+1}$.  
To this end, a tentative  policy is proposed to solve Problem (P7). The principles behind this policy is that $\min_{k\in\bar\cK_{n+1}}T^{(d)}_k$ only takes values from a finite discrete set 
$\left\{T^{(d)}_k\Big|k\in\cF\right\}$ such that we can judge the feasibility of each value of  $\min_{k\in\bar\cK_{n+1}}T^{(d)}_k$ sequentially. 
Specifically, with given $\Pi$, define a set $\cG$ containing all the values of unscheduled tasks' deadlines, i.e., $\cG=\left\{T^{(d)}_k\Big|k\in\cF\right\}$. At each time, we check whether the unscheduled task with the earliest deadline can be added in $\bar\cK_{n+1}$. Set $T^{(d)}$ as the latest  completion instant for the $\left(n+1\right)$-th batch, that is  $T^{(d)}=\min\left\{\min\left\{i|i\in\cG\right\},\min_{k\in\cS_{n+1}}T^{(d)}_k,t_{n+2}\right\}$. Thus,  the optimal batch starting instant is given as  $\bar{t}^*_{n+1}=T^{(d)}-d\left(\Pi+|\cS_{n+1}|\right)$. 
Let  ${\tilde \cS}=\left\{k\Big| T^{(a)}_k<\bar{t}^*_{n+1},T^{(d)}_k\geq T^{(d)},k\in\cF\right\}$ denote the set of tasks that not only satisfy the uploading causality constraint but also their deadlines is no earlier than $T^{(d)}$. This indicates that the tasks in ${\tilde\cS}$ can meet their deadlines requirements even with an earlier deadline. 
Next, we judge whether there exist $\Pi$ new-scheduled tasks in $\tilde\cS$ satisfying bandwidth constraint. Sort the tasks belonging in $\tilde\cS$ 
in ascending order according to the value of  minimum bandwidth required, i.e., $\frac{ \tilde \ell_k}{r_k\left(\bar t_{n+1}-\max\left\{ T^{(a)}_k,\bar t_n\right\}\right)}$ and  assemble the first $\Pi$ tasks in set $\bar \cK_{n+1}$. If the total bandwidth of tasks in $\bar\cK_{n+1}\cup\cS_{n+1}$  is no larger than $\sum_{i=1}^{n+1}\sum_{k\in\cS_i}B_k$ and $\max\left\{\max_{k\in\bar\cK_{n+1}\cup\cS_{n+1}} T^{(a)}_k, \bar t_n+d(|\cS_{n}|)\right\}$ is smaller than $\bar t_{n+1}$, i.e., bandwidth and  task-causality constraints in Problem (P7) are satisfied,  Problem (P7) under current $T^{(d)}$ is feasible, thus making Problem (P7) with current $\Pi$ feasible. Hence, we set $\Pi_{\text{low}}=\Pi$. Otherwise, it is infeasible with current $T^{(d)}$ indicating that task $k=\arg\min\left\{i\in\cG\right\}$ cannot be scheduled in current batch. In this case, we should delete the minimum value in $\cG$. Then update  $T^{(d)}=\min\left\{\min\left\{i|i\in\cG\right\},\min_{k\in\cS_{n+1}}T^{(d)}_k,t_{n+2}\right\}$ and repeat the above steps until the number of  elements in $\tilde \cS$ is less than $\Pi$. This implies that problem with $|\bar\cK_{n+1}|=\Pi$ is infeasible and we set $\Pi_{\text{up}}=\Pi$. At each checkpoint $t_n$ for $n=1,\cdots,N-1$, we solve Problem (P7) and update  $\cF$ as $\cF\setminus\bar\cK_{n+1}$, $\cS_{n+1}$ as $\cS_{n+1}\cup\bar\cK_{n+1}$  until $\cF$ is empty. 

\begin{algorithm}[t]
	\begin{small}
		\caption{Spectrum-Hole Allocation Algorithm}
		\label{alg_reallocation}
		
		Initialize $ \{ T^{(a)}_k\}$,  $\{T^{(d)}_k\}$, $\{A_k\}$,  $\{B_k\}$, $\cF$, $\bar t_1=t_1$, and $\{\cS_n\}$.

		\For{$n=1,\cdots,N-1$}{Set  $\Pi_{\text{up}}=|\cF|$, and $\Pi_{\text{low}}=0$.  
			
		\Repeat{$\Pi_{\text{up}}-\Pi_{\text{low}}\leq1$}{	Set $\Pi=\lfloor\left(\Pi_{\text{up}}+\Pi_{\text{low}}\right)/2\rfloor$. 
			
			Set $\cG=\left\{T^{(d)}_k\Big|k\in\cF\right\}$.

			\Repeat{$|\tilde\cS|<\Pi$}{
				Set $T^{(d)}=\min\{\min\left\{i|i\in\cG\right\},\min_{k\in\cS_{n+1}}T^{(d)}_k,t_{n+2}\}$. 
				
				Set $\bar t_{n+1}=T^{(d)}-d\left(\Pi+|\cS_{n+1}|\right)$. 
				
				
				Set  $\tilde{\cS}=\left\{k\Big| T^{(a)}_k<\bar t_{n+1},T^{(d)}_k\geq T^{(d)},k\in\cF\right\}$

				Sort the elements in $\tilde\cS$ in an ascending order according to $\frac{\tilde \ell_k}{r_k\left(\bar t_{n+1}-\max\{ T^{(a)}_k,\bar t_n\}\right)}$. 
				
				Assemble the first $\Pi$ elements of $\tilde\cS$ in set $\bar\cK_{n+1}$. 
				
				\eIf{the  bandwidth and task-causality constraints are satisfied}{
					$\Pi_{\text{low}}\leftarrow\Pi$. 
					
					Break.
				}{
					$\cG\leftarrow\cG\setminus\min\{i|i\in\cG\}$}
				
			}
			\If{$|\tilde\cS|<\Pi$}{
				$\Pi_{\text{up}}\leftarrow\Pi$.}
			
		}
		\eIf{$\Pi_{\text{up}}$ is feasible}{$\Pi^*=\Pi_{\text{up}}$}{$\Pi^*=\Pi_{\text{low}}$}	
		
		Update $\cF\leftarrow\cF\setminus\bar\cK_{n+1}$, $\cS_{n+1}\leftarrow\cS_{n+1}\cup\bar\cK_{n+1}$.}
		
	\end{small}  
\end{algorithm}

The detailed steps for spectrum-hole allocation scheme is summarized in Algorithm~\ref{alg_reallocation}  whose computational complexity lies in solving Problem (P7) with given $\Pi$.  
Specifically, with given $\Pi$, the complexity for (P7) is $\mathcal{O}\left(K^2\right)$. Since we have to solve (P7) with each $\Pi$ and $n$, the total computational  complexity for  the proposed spectrum-hole allocation algorithm is estimated as $\mathcal{O}\left((N-1)K^2\log_2(K)\right)$. 
It should be noted that Algorithm~\ref{alg_reallocation} can obtain a globally optimal solution for Problem (P6) with low complexity.

\section{Extensions and Discussion}

\subsection{Online Task Admission}
During the process of task uploading and inference, new tasks may arrive and submit service requests \cite{9812874}. In this scenario, the proposed Algorithm~\ref{alg_reallocation} can be modified to support online admission of new tasks to leverage spectrum holes.  
Specifically, the proposed Algorithm~\ref{alg1} is executed as batching initialization for existing tasks. Then new tasks arriving during the inference process are first stored locally at devices. 
Similar to Algorithm~\ref{alg_reallocation}, at each checkpoint $t_n$, we update the set of active tasks,  $\cF$, to include selected new tasks that are deemed feasible for successful execution using spare resources.  
To this end,  Algorithm~\ref{alg_reallocation} can be executed again to update the resource allocation to accommodate the new tasks. 

\subsection{Frequency-Selective Channels}
The current assumption of frequency non-selective can be relaxed as follows. A frequency selective channel  can be partitioned using  orthogonal frequency division multiplexing  (OFDM) into sub-channels with heterogeneous gains. Following  \cite{7762913,6025328},  new indicator variables can be introduced to denote the association between sub-channels and tasks. Then the throughput maximization problem can be  formulated as a MINLP problem containing two kinds of binary optimization variables for  sub-channel allocation and  task-batch association, respectively. Despite being more complex, the problem can be solved efficiently using conventional MINLP methods such as convex relaxation and branch-and-bound, or latest approach using machine learning (see e.g., \cite{9449944}).

\section{Simulation Results}
\subsection{Simulation Settings}
The default settings are as follows. 
There are  $K=100$ devices, with task arrivals uniformly and independently generated in the time interval of $[0, 1]$ s. The size of feature vectors is set as $10$ KBytes. The delay requirements of tasks follow the uniform distribution in $[0.05,2]$ s. The inference latency profile  with respect to the batch size as reported in \cite{franklin2019nvidia} is adopted, which is generated from a  ResNet-50 model implemented on JETSON TX1 and the ImageNet dataset.  The channel gains between devices and server follow independent  Rayleigh fading with the average power loss being $10^{-3}$. The transmit \emph{signal-to-noise ratio} (SNR) of devices is set as $20$ dB. The constant $\delta$ in $\ell_0$-norm approximation \eqref{eq_ell0} is $10^{-15}$. 
The following schemes are considered in performance  comparison:
\begin{itemize}
	\item \textit{Proposed Algorithm:} See Algorithm~\ref{alg1}. 
	\item \textit{Equal Bandwidth Allocation Scheme:} The total bandwidth is evenly allocated to devices while task scheduling follows Algorithm~\ref{alg1}.
	\item \textit{Spectrum-Hole Allocation Scheme:}  Algorithm~\ref{alg1} enhanced with spectrum-hole allocation using  Algorithm~\ref{alg_reallocation}. 
	\item \textit{Greedy Batching Scheme:} Upon finishing executing the previous  batching, the server greedily assembles all tasks that arrived during the previous batch into a new  batch and makes inferences on them \cite{9355312}. When the inference is finished, those tasks that do not satisfy deadline requirements are discarded.   
	\item \textit{Single Batch Scheme:} 
    The optimal single batch scheme in \cite{10038543} is modified for asynchronous task arrivals. In particular, the optimal batch starting instant is determined using an algorithm similar to Algorithm~\ref{alg_reallocation}. 
\end{itemize}
The performance metric of task completion rate is defined as the ratio between completed tasks and all  tasks. 
Note that the metric measures system throughput.

\subsection{Effect of Task Number}

\begin{figure}[!t]
	\centering
		\includegraphics[width=0.7\textwidth]{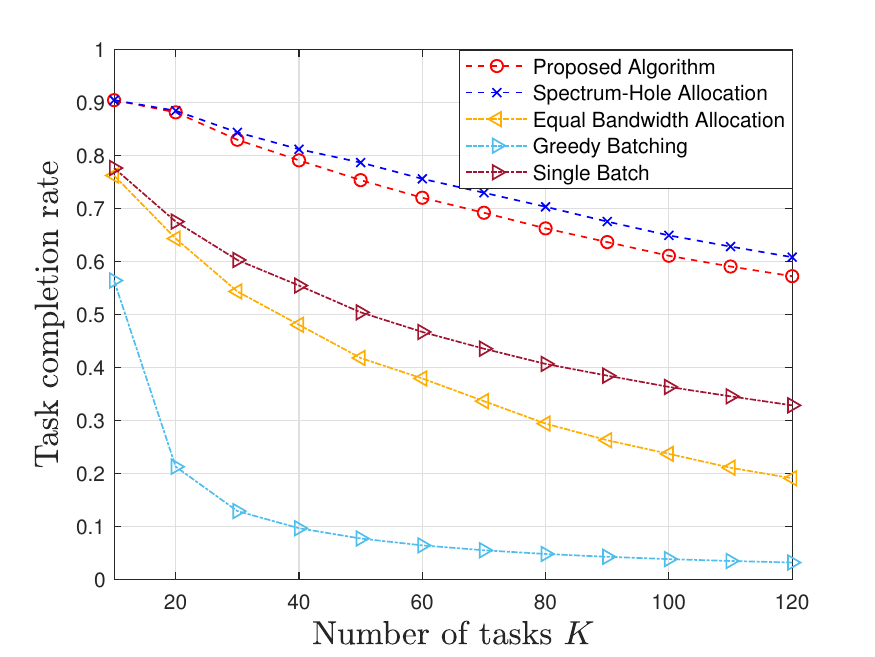} 
		\caption{Task completion rate versus number of tasks.}\label{Throu_K} 
	\end{figure}
 \begin{figure}
		\centering
		\includegraphics[width=0.7\textwidth]{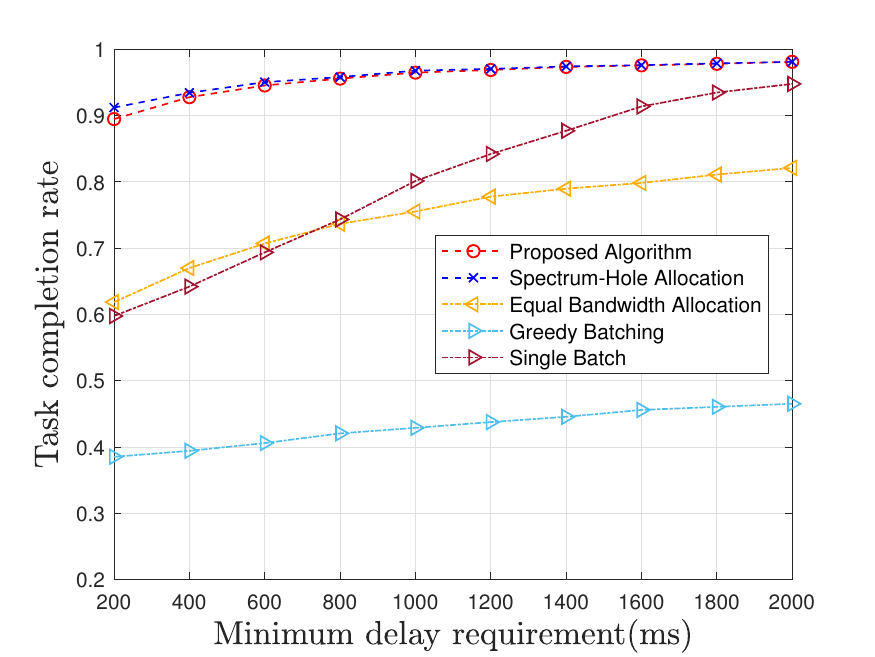}
		\caption{Task completion rate versus the minimum delay requirement.}\label{Throu_delay}
\end{figure}

Fig.~\ref{Throu_K} compares task completion rates between different schemes for a varying number of tasks. The proposed JBAS scheme and its enhanced version with spectrum-hole exploitation achieve the highest  rates.  This shows the advantages of jointly optimizing batching, scheduling, and bandwidth allocation  
 so as to accommodate the heterogeneity of task arrivals and deadlines. In contrast, the three benchmarking schemes are less effective in accounting for the different delay requirements of tasks and balancing the tradeoff between batch size and batch startup instants. 
 As a result, they suffer loss on  system throughput that is larger as the number of tasks grows. 
 On the other hand, we can observe that the Spectrum-Hole Allocation Scheme can enhance the throughput of  the Proposed Scheme by an average of $2.8\%$. Furthermore, as observed from Fig.~\ref{Throu_K}, the task completion rates  decrease as the number of tasks grows. This indicates that the limited communication and computation source leads to an increasing slower in the number of completed tasks as the total number grows. 

\subsection{Effect of Delay Requirement}
To investigate the effect of delay requirements on system throughput, we vary the minimum delay requirement from $50$ to $1450$ ms while the maximum delay requirement is fixed at $2000$ ms. 
The curves of task completion rate  versus the minimum delay requirements  are depicted in Fig.~\ref{Throu_delay}. One can observe that the  task completion rates  of all schemes gradually increase as the minimum delay requirement relaxes.  The reason is that less bandwidth is required for each task for uploading and the server has more computation time.  From Fig.~\ref{Throu_delay}, we can observe that the throughput improvement of the Spectrum-Hole Allocation Scheme on top of the Proposed Scheme reduces from $1.7\%$ to zero as the minimum delay increases from $200$ to $2000$ ms. This  can be explained by that as the minimum delay increases, the radio resource constraints are relaxed and the communication bottleneck is dominated by the computation counterpart. 
In contrast, the  throughput improvement of  Spectrum-Hole Allocation Scheme is more significant in spectrum constrained scenarios, (i.e., a large number of tasks and tight delay requirements) as shown in  Fig.~\ref{Throu_K} and Fig.~\ref{Throu_delay}. Furthermore, as the minimum delay increases, the throughput of the One Batch Scheme improves rapidly as the loss on synchronizing tasks' starting time reduces.

\begin{figure}[!t]
		\centering
		\includegraphics[width=0.7\textwidth]{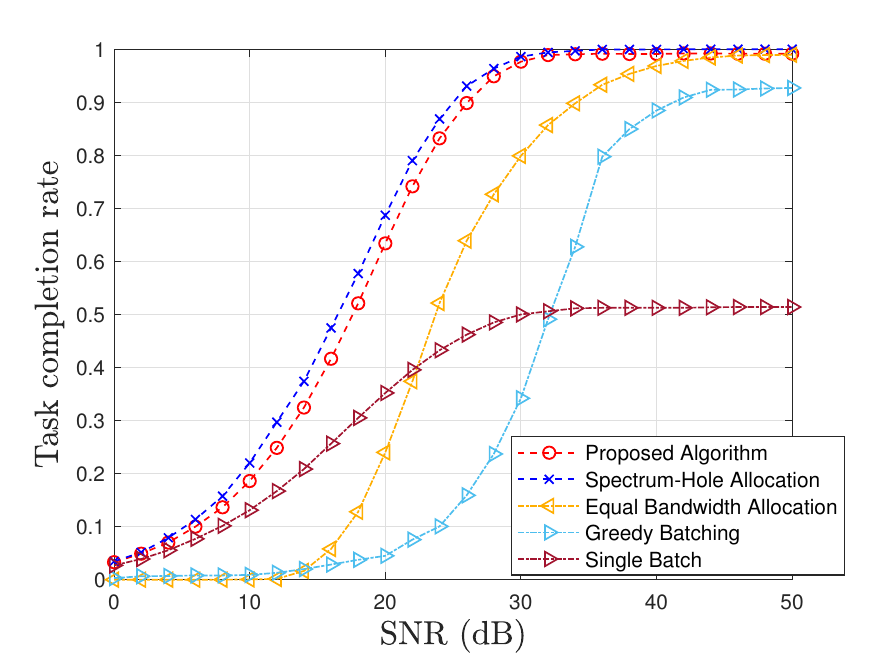} 
		\caption{Task completion rate versus transmit  SNR.}\label{Throu_SNR} 
	\end{figure}

\subsection{Effect of Transmit SNR}
In Fig.~\ref{Throu_SNR}, the curves of task completion rate performance versus transmit SNR are plotted. 
As the transmit SNR increases, the task completion rate first improves rapidly and then saturates. The early rapid improvement reflects the overcoming of the communication bottleneck. 
As the SNR is further increased,  the  bandwidth constraint becomes inactive, leading to throughput saturation. 
In this operation regime, the computation bottleneck dominates and limits system throughput. 
 One can observe that with sufficiently large SNR (e.g., $50$ dB), the Proposed, Spectrum-Hole Allocation, and Equal Bandwidth Allocation Schemes can complete almost all tasks, while the Greedy Batching Scheme only reaches $92\%$ task completion rate, which verifies the need of batching optimization. 
On the other hand, the Single Batch Scheme performs worst at a large SNR, i.e., less than  $48\%$ task completion rate, indicating the importance of multiple batches for asynchronous tasks arrivals.  

\vspace{-5mm}
\subsection{Effect of Batch Number}
 \begin{figure}
		\centering
		\includegraphics[width=0.7\textwidth]{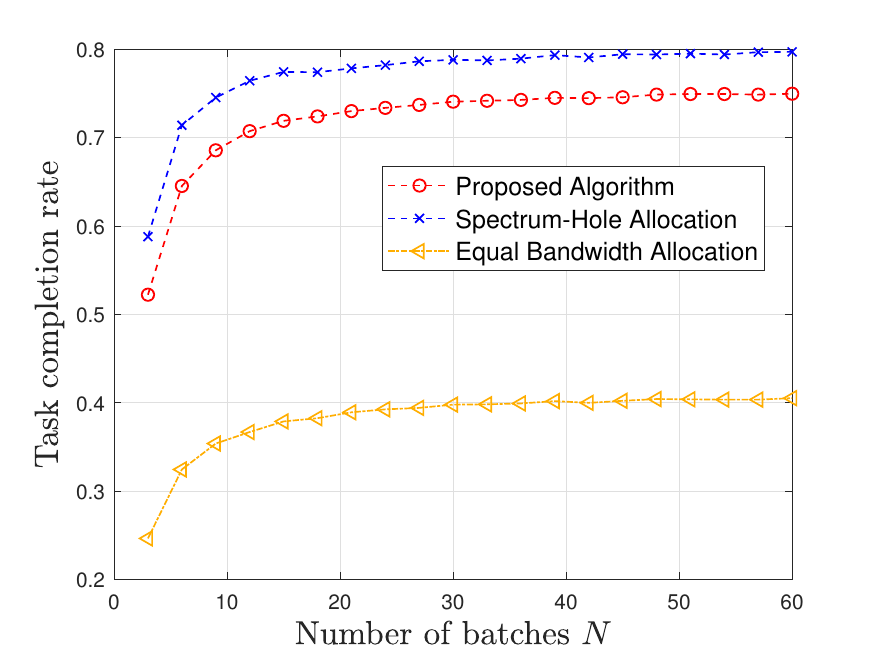}
		\caption{Task completion rate versus number of batches with $K=60$.}\label{Throu_Batchnum}
\end{figure}

Fig.~\ref{Throu_Batchnum} shows the curves of task completion rate  versus number of batches from solving Problem (P3). 
As can be observed, as the number of batches grows, the system throughput increases and then saturates as the effective number of batches, namely the non-empty ones, converges to a fixed value. 
Besides, the proposed Scheme and Spectrum-Hole Allocation Scheme achieve the throughput improvement of $39.34\%$ and $47.50\%$, respectively, compared with Equal Bandwidth Allocation Scheme when $K=60$.

\begin{figure*}[!t]
	\centering
	\includegraphics[width=0.7\textwidth]{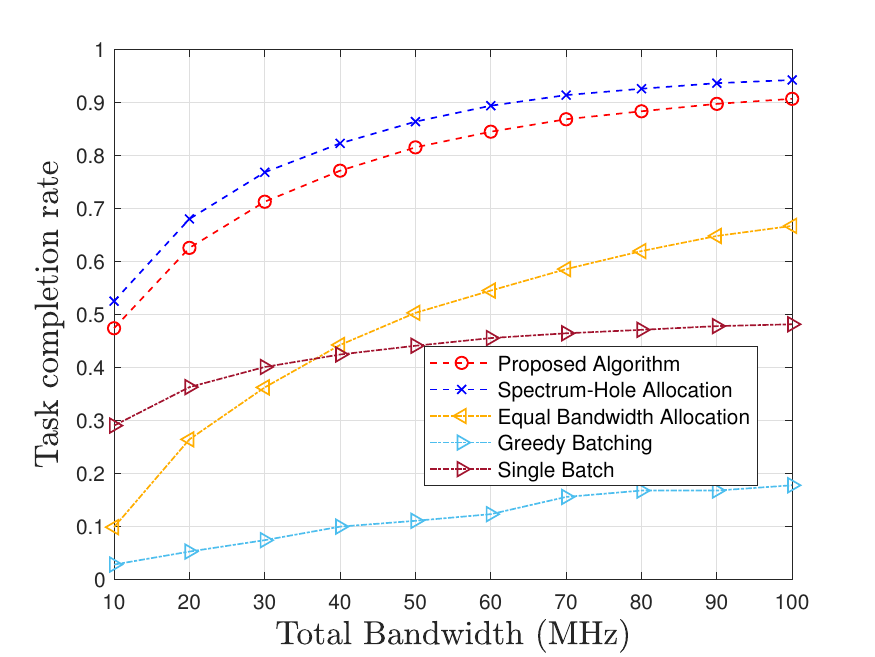}
	\caption{Task completion rate versus total bandwidth.}\label{Throu_BW}
	\centering
\end{figure*}

\subsection{Effect of Bandwidth}
In Fig.~\ref{Throu_BW}, the curves of  task completion rate versus total  bandwidths are plotted.  As can be seen, as the bandwidth increases, the throughput of all schemes increases due to the communication resources getting more abundant. The Proposed Scheme achieves $29.80\%$, $66.48\%$, and $35.89\%$ throughput gains compared with Equal Bandwidth Allocation Scheme, Greedy Batching Scheme, and One Batch Scheme, respectively. Moreover, the Spectrum-Hole Allocation Scheme can further improve the throughput by $4.61\%$.

\section{Conclusion}
In this paper, we have presented a JBAS  framework for high-throughput multiuser edge AI in the practical scenarios with heterogeneous task arrivals and deadlines. 
  The number of batches, batch startup instants, task-batch association, as well as bandwidth allocation have been jointly optimized to maximize the system throughput.
Moreover, spectrum holes have been exploited to further increase the throughput.  
We find that judiciously assembling tasks into multiple batches is  important to ensure high throughput in practice.  
However, the communication model considered in this paper is simple for the sake of tractability.  For future works, it is promising to integrate batching with advanced  transmission techniques such as \emph{non-orthogonal multiple access} (NOMA)  and \emph{multiple-input multiple-output} (MIMO). In another interesting direction, multiple-cell edge AI with batching couples communications in different cells and computation at different servers and hence is more challenging to design.

\appendix
\subsection{Proof of Theorem~\ref{theorem_3}} \label{proof_theorem_3}
Denote the optimal solution of Problem (P2) by $\left(\left\{t^{*}_n\right\},\left\{\pi_{k,n}^*\right\},N^* \right)$. We first prove that for the optimal solution of Problem (P2), the optimal value of (P3) is no less than that of (P2).  Consider the following two cases:

\textit{Case 1: $N^*=K$.} If the optimal $N^*=K$,  $\left(\left\{t^{*}_n\right\},\left\{\pi_{k,n}^*\right\} \right)$ is feasible to Problem (P3) since that (P3) is the case when $N=K$. 

\textit{Case 2: $N^*<K$.} In this case, introducing new variables  $t^{^*}_{n}=\Xi$ $\left(\forall n\in\{N^*+1,\cdots,K\}\right)$ and $\pi_{k,n}^*=0$ $\left(\forall k\in\cK, \forall n\in\{N^*+1,\cdots,K\}\right)$. Then combing the optimal solution of (P2) and the new introduced variables, the constructed variables $\left(\left\{t^{*}_n\right\}_{\forall n\in\cK},\left\{\pi_{k,n}^*\right\}_{\forall k\in\cK,\forall k\in\cK}  \right)$ satisfy all the constraints in Problem (P3). Moreover, the optimal value of (P2) is equal to the value of (P3). Therefore, the optimal value of (P3) is no less than that of (P2). 

Next, since the optimal solution of (P3) always satisfies the constraints of (P2). Hence, the feasibility of (P3) is included in that of (P2). In other words, the optimal solution of (P3) is feasible to (P2). Thus, the optimal value of (P2) is no less than that of (P3). 

Combining that the optimal value of (P2) is no less than and also no larger than that of (P3), we can conclude that Problem (P3) is equivalent to Problem (P2).  

\subsection{Proof of Theorem~\ref{theorem_2}} \label{proof_theorem_2}
Startup time $t_n$ is lower bounded by the task-arrival time instants of its associated tasks and upper bounded by deadlines and start time of the next batch $t_{n+1}$. Since the allocated bandwidth $B_k$ decreases with $t_n$, we should set $t_n$ as large as possible in order to satisfy the bandwidth constraint. Through solving $t_{n}$ sequentially from $n=N$ to $n=1$, we can obtain the optimal solution of $t_{n}$. Specifically, for the $N$-th batch,
we consider the following two cases: 1) If the $N$-th batch is non-empty, i.e., $\sum_{k=1}^K\pi_{k,N}>0$, the time instant that the $N$-th batch finishes its inference $t_N+d_N\left(\pi_N\right)$ is upper bounded by the deadlines of its associated tasks. Hence, we should let $t_N=\min_{ k\in\cK_N}T^{(d)}_k-d_N\left(\pi_N\right)$. 2) If the $N$-th batch is empty, i.e., $\sum_{k=1}^K\pi_{k,N}=0$, we should set $t_N$ as large as possible such that it will not affect the value of $t_{N-1}$. Without loss of generality, we set $t_N=\Xi$. Subsequently, consider the $(N-1)$-th batch. Similarly, two cases are considered. If it is non-empty, time instant $t_{N-1}+d_{N-1}\left(\pi_{N-1}\right)$ is restricted not only by the deadlines of its associated tasks but also by the startup time of the $N$-th batch. Therefore,  $t_{N-1}$ is set to $\min\left\{\min_{k\in\cK_{N-1}}T^{(d)}_k,t_{N}\right\} -  d_{N-1}\left(\pi_{N-1}\right)$. If the $(N-1)$-th batch is empty, in order to mitigate the influence on the startup time of the $(N-2)$-th batch, we set $t_{N-1}=t_N$. Following this procedure until the startup time of the first batch $t_1$ is obtained, completing the proof.

\subsection{Proof of Proposition~\ref{proposition_1}} \label{proof_proposition_1}
	First, when $\Pi<\Pi^*$, we have $d(\Pi+|\cS_{n+1}|)\leq d(\Pi^*+|\cS_{n+1})$ since the inference delay function is non-decreasing.  Therefore, let $\bar t_{n+1}=\bar{t}_{n+1}^*$ and $\bar\cK_{n+1}$ be an arbitrary subset of $\bar\cK_{n+1}^*$ with the size of $\Pi$. We can deduce $\bar{t}_{n+1}$ and $\bar\cK_{n+1}$ satisfy all the constraints of (P7). Thus, Problem (P7) is feasible. 

Subsequently, when  $\Pi>\Pi^*$, Problem (P7) is always feasible, otherwise, the optimal $|\bar\cK_{n+1}^*|$ of (P6) is larger than $\Pi^*$, which contradicts that $|\bar\cK_{n+1}^*|=\Pi^*$.  That completes the proof.

\bibliographystyle{IEEEtran}
\bibliography{Ref}

\end{document}